\documentclass[11pt]{article}
\usepackage{rheaj}
\usepackage{setspace}
\usepackage{xparse}
\usepackage{tabularx}
\usepackage{enumitem}

\newcommand{\RegionGrow}{\textsc{RegionGrowing}}
\newcommand{\PathCutting}{\textsc{PathCutting}}
\newcommand{\curr}{\textnormal{curr}}
 
\begin{document}
\title{
  Node-Weighted Multicut in Planar Digraphs\thanks{Siebel School of Computing and Data
    Science, University of Illinois, Urbana-Champaign, Urbana, IL
    61801. {\tt \{chekuri,rheaj3\}@illinois.edu}. Supported in part by
    NSF grant CCF-2402667.}}
\author{Chandra Chekuri \and Rhea Jain}
\date{}
\maketitle

\begin{abstract}
  Kawarabayashi and Sidiropoulos \cite{KS22} obtained an
  $O(\log^2 n)$-approximation algorithm for Multicut in \emph{planar
    digraphs} via a natural LP relaxation, which also establishes a
  corresponding upper bound on the multicommodity flow-cut gap. Their
  result is in contrast to a lower bound of $\tilde{\Omega}(n^{1/7})$
  on the flow-cut gap for general digraphs due to Chuzhoy and Khanna
  \cite{ChuzhoyK}. We extend the algorithm and analysis
  in \cite{KS22} to the \emph{node-weighted} Multicut problem. Unlike
  in general digraphs, node-weighted problems cannot be reduced to
  edge-weighted problems in a black box fashion due to the planarity
  restriction.  We use the node-weighted problem as a vehicle to
  accomplish two additional goals: (i) to obtain a deterministic algorithm
  (the algorithm in \cite{KS22} is randomized), and (ii) to simplify
  and clarify some aspects of the algorithm and analysis from
  \cite{KS22}. The Multicut result, via a standard technique, implies
  an approximation for the Nonuniform Sparsest Cut problem with an
  additional logarithmic factor loss.
\end{abstract}
\section{Introduction}
\label{sec:intro}
The well-known $s$-$t$ maxflow-mincut theorem is a cornerstone result
in discrete and combinatorial optimization with many consequences and
applications \cite{Schrijver-book,AMO-book,Williamson-book}.
Multicommodity flows, where multiple demand pairs aim to route flow
simultaneously, and related cut problems, are also of significant
importance and are extensively
studied. Unlike in the single-commodity setting, flow-cut equivalence
does not hold in these more general settings. Motivated by
applications to problems such as Multicut, Sparsest Cut, routing
problems, and others, there is extensive work in establishing a
quantitative relationship between multicommodity flows and cuts --- we
refer to this broad set of results as flow-cut gaps.  
Leighton and
Rao \cite{LeightonR}, in their seminal work,
established a tight $\Theta(\log n)$ flow-cut gap for the uniform
sparsest cut problem in undirected graphs and used it in a host of
applications. Garg, Vazirani, and Yannakakis \cite{GVY} used the
region growing technique from \cite{LeightonR} to obtain a tight
$\Theta(\log n)$-approximation and flow-cut gap for the Multicut
problem in undirected graphs. Linial, London, and Rabinovich
\cite{LLR} and Aumann and Rabani \cite{AumannR98} obtained a tight
$\Theta(\log n)$-approximation and flow-cut gap for non-uniform
sparsest cut via a connection to $\ell_1$-embeddings.

Planar and
minor-free graphs are important special classes of graphs for
which improved results are known. Klein, Plotkin and Rao \cite{KPR}
obtained an $O(1)$-approximation and flow-cut gap for Multicut and
uniform sparsest cut. The well-known GNRS conjecture of Gupta et al.
\cite{GNRS} states that the flow-cut gap for non-uniform sparsest cut
is $O(1)$ for planar and minor-free graphs. Rao \cite{Rao} showed
that it is $O(\sqrt{\log n})$. These are but a small number of
important results on cut problems; we elaborate more in \Cref{sec:rel_work}.

In this paper, we are concerned with multicommodity flows and cuts in
\emph{directed} graphs. In particular, we are motivated by
approximation algorithms for the Multicut problem, which we state formally.

\medskip
\noindent {\bf Multicut.} Given a directed graph $G=(V,E)$ with
non-negative edge-weights $c(e)$ for all $e \in E$ and $k$ source-sink
pairs $(s_1,t_1), \ldots, (s_k,t_k)$, find a min-cost subset
$E' \subseteq E$ of edges such that there is no path from $s_i$ to
$t_i$ for any pair $(s_i,t_i)$ in the graph $G-E'$. In
\emph{node-weighted} Multicut, each node $v$ has a cost/weight $c(v)$,
and the goal is to remove a min-cost subset $V' \subseteq V$ such that
$G-V'$ has no $s_i$-$t_i$ path for any of the input pairs. We assume
that sources and sinks have infinite cost so that they cannot be
removed, otherwise we can add dummy terminals to satisfy this assumption.

There is a natural LP relaxation for Multicut (see \Cref{sec:results} and \Cref{sec:prelim}), 
and the dual of the LP
relaxation corresponds to the maximum throughput multicommodity flow
problem.  Thus, the integrality gap of this LP is the same as the
flow-cut gap. Closely related to the Multicut problem is the
Non-uniform Sparsest Cut problem, which we abbreviate to Sparsest Cut
in this paper.

\medskip
\noindent
{\bf Sparsest Cut.}  Given a directed graph $G=(V,E)$ with non-negative
edge-weights $c(e)$ for all $e \in E$ and $k$ source-sink pairs $(s_1,t_1),
\ldots, (s_k,t_k)$, each with a non-negative demand $D_i$, find a
subset $E' \subseteq E$ of edges to minimize the ratio
$\frac{c(E')}{\text{dem}(E')}$; here $\text{dem}(E')$ is the sum over
demands of all of the given pairs that are separated in $G - E'$.
In the \emph{node-weighted} Sparsest Cut problem each node $v$ has a
cost/weight $c(v)$, and the goal is to remove
a subset $V' \subseteq V$ to minimize $\frac{c(V')}{\text{dem}(V')}$,
where $\text{dem}(E')$ is the demand of the pairs separated in
$G-V'$. 

\begin{remark}
  It is common to define sparsest cut in undirected graphs as finding
  a vertex subset $S$ with smallest ratio $c(\delta(S))/\text{dem}(S)$
  where $\text{dem}(S)$ is the demand crossing the set $S$. The edge-based 
  definition and the vertex-subset-based definition coincide in
  undirected graphs, but they are not the same in directed graphs. 
\end{remark}

\medskip In contrast to undirected graphs, where tight
logarithmic factor flow-cut gaps are known for both Multicut and
Sparsest Cut, the corresponding flow-cut gaps in directed graphs are
significantly larger.  Saks, Samorodnitsky and Zosin \cite{SaksSZ}
showed that the Multicut flow-cut gap in directed graphs can be
$k(1-o(1))$ where $k$ is the number of pairs --- note that an upper
bound of $k$ is easy.  In their construction $k = O(\log
n)$. Subsequently, Chuzhoy and Khanna \cite{ChuzhoyK} showed a lower
bound of $\tilde{\Omega}(n^{1/7})$, and also established almost
polynomial-factor hardness of Multicut in directed graphs.  

More recently, Kawarabayashi and Sidiropoulos \cite{KS22} obtained an interesting
positive result in \emph{planar} digraphs. A planar digraph is
a directed graph obtained from a planar multi-graph by orienting
each of its undirected edges.

\begin{theorem}[\cite{KS22}]
  There is an $O(\log^2 n)$-approximation for Multicut in planar
  digraphs via the natural LP relaxation.
\end{theorem}

They use a known reduction from Sparsest Cut to Multicut via the
LP \cite{Kahale,Shmoys97,AgarwalAC07} to obtain the following corollary.

\begin{corollary}[\cite{KS22}]
  There is an $O(\log^3 n)$-approximation for Sparsest Cut in planar
  digraphs via the natural LP relaxation.
\end{corollary}

The preceding results are likely to extend to digraphs over minor-free
graphs. One would then obtain good algorithmic
results for various cut problems in large families of digraphs that arise
in several applications. These results also provide impetus to understand tight
flow-cut gaps for planar digraphs and can potentially lead to other
applications. For instance, some of the ideas in \cite{KS22} have
inspired recent work in network design on planar digraphs
\cite{FriggstadM23,ChekuriJKZZ24,ChekuriJ25}.

\subsection{Result and Technical Overview}
\label{sec:results}

The results in \cite{KS22} are for the edge-weighted versions of
Multicut and Sparsest Cut. It is easy to see that edge-weighted
problems reduce to the node-weighted case. In general digraphs, one can
do a simple and standard reduction to reduce node-weighted problems to
edge-weighted problems as follows. Split a node $v$ into two nodes
$v_{in}$ and $v_{out}$ and add the edge $(v_{in}, v_{out})$ with cost
$c(v)$. For each original edge $(u,v)$ we add the edge
$(v_{out},u_{in})$ with cost $\infty$.  This reduction, however, does
\emph{not} preserve planarity and hence we cannot use \cite{KS22} as a
black box.  We note that in undirected graphs, there is no reduction
from node-weighted to edge-weighted
Multicut and Sparsest Cut. In fact, the $\ell_1$-embedding
technique for Sparsest Cut is not the right tool for node-weighted
problems; known results are based on line embeddings and
various other refinements \cite{FeigeHL08, BrinkmanKL07,CKRV,LeeMM15},
and our understanding
of node-weighted problems is less extensive. A natural question in
light of these considerations is whether the result of Kawarabayashi
and Sidiropoulos \cite{KS22} extends to the node-weighted setting. In
this paper we show that this is indeed the case via the following
theorem.

\begin{theorem}
\label{thm:main}
  There is an efficient deterministic $O(\log^2 n)$-approximation 
  algorithm for node-weighted Multicut in directed planar graphs via
  the natural LP relaxation.
\end{theorem}

Via the standard reduction from Sparsest Cut to Multicut via the LP
that we alluded to earlier, we obtain the following.

\begin{corollary}
  There is an efficient deterministic $O(\log^3 n)$-approximation for
  the node-weighted Sparsest Cut problem in directed planar graphs via
  the natural LP relaxation.
\end{corollary}

Node-weighted problems are a natural and well-motivated generalization
of edge-weighted problems and thus our work yields a useful extension
of the results from \cite{KS22}.  We also believe the ideas can extend to
the polymatroid network model, which is a further generalization of
node-weights with a variety of applications --- we omit a
detailed discussion due to the technical nature of the model and refer
the interested reader to past and recent work \cite{CKRV, LeeMM15,
ChenOT25, ChekuriL24}. 
Our second motivation is to clarify the core
techniques in \cite{KS22} to enable future improvements
and extensions. 

We sketch the approach of \cite{KS22} and explain
along the way how we obtain a deterministic algorithm that
extends to node-weights. We start with the standard LP relaxation for
edge-weighted Multicut in digraphs. For each edge $e \in E$ there is a
variable $x_e \in [0,1]$ indicating (in the integer programming
setting) whether $e$ is cut. For $i \in [k]$, let $\calP_i$ denote the set of all
$s_i$-$t_i$ paths. To separate the pair, each path $p \in P_i$ must
have an edge on it deleted. This hitting set LP relaxation can be
written with an exponential number of constraints as given below.

\begin{align*}
  \min \sum_{e \in E} c(e) x_e& \\
  \text{subject to } \sum_{e \in p} x_e &\geq 1 &\forall p \in \calP_i, \forall i \in [k] \\
  x_e &\geq 0 &\forall e \in E
\end{align*}

The LP can be solved in polynomial-time since the separation oracle
corresponds to the shortest path problem. The dual of the preceding LP
relaxation is easily seen to be the LP for the maximum throughput
multicommodity flow for the given pairs. As remarked previously,
establishing an integrality gap of $\alpha$ on the Multicut LP is the
same as establishing a flow-cut gap of $\alpha$. The algorithm of
\cite{KS22} can be viewed as a \emph{randomized} rounding algorithm that, 
given a feasible solution ${\bf x}$ to the LP relaxation, outputs a
feasible multicut $E'$ for the given instance such that
$\Pr[e \in E'] \le \alpha \cdot x_e$. In \cite{KS22}, the authors show
that $\alpha = O(\log^2 n)$ suffices.
It is helpful to think of the $x_e$ values as lengths; with this
viewpoint, $E'$ is feasible if all pairs of distance $\geq 1$ are 
separated in $G \setminus E'$.

\medskip

\noindent {\bf Randomized region growing and divide and conquer.}
The approach of \cite{KS22} is based on a divide and conquer approach
that relies on the fact that planar digraphs have rooted shortest-path 
separators; this was implicitly proven by Lipton and Tarjan 
\cite{LiptonTarjan79} and made explicit by Thorup \cite{Thorup}.
The result states that given any \emph{undirected}
planar graph with a spanning tree $T$ rooted at $r$, there are three
nodes $u_1, u_2, u_3$ such that if we were to remove the unique tree
paths connecting $r$ to $u_i$ for $i = 1,2,3$ (the union of these
paths is called the \emph{separator}), each resulting connected
component would have at most $n/2$ nodes. This lends itself nicely to
an algorithmic framework in which we solve the problem on the
separator and then recursively handle each connected component
individually.  The planar separator result can be applied to directed
graphs by considering the underlying undirected graph.  In this case,
the removal of the separator results in weakly connected components
which each have at most $n/2$ nodes.  However, the separator may not
consist of paths in the directed sense.  Thus, the first step of
\cite{KS22} is to preprocess the graph via a decomposition into
layers, each of which has a root that can either reach or be reached
by all nodes in its layer. This allows the use of the separator result
on a spanning directed in-tree or out-tree, ensuring that each of the
three paths on the separator is a directed path.  The main technical
ingredient after these steps (which are inherited to some extent from
the data structure work in \cite{Thorup}) is to solve the problem on
the separator; this corresponds to cutting long $s_i$-$t_i$ paths that
intersect the separator. For this cutting procedure, \cite{KS22}
describes a randomized rounding algorithm which can be viewed as
Bartal's ball cutting procedure \cite{Bartal} applied to directed
graphs. We observe that the authors use Bartal's scheme because they
need to apply it multiple times and the memory-less property of the
scheme is crucial.

\medskip
\noindent {\bf A deterministic algorithm.} We follow the same
high-level approach. In order to generalize to node-weights and
obtain a deterministic algorithm, we 
replace the use of Bartal's ball growing procedure with 
a standard region growing
lemma of Garg, Vazirani and Yannakakis \cite{GVY}. This lemma is used in several
Multicut variants, including in directed graphs, and is designed
precisely for sequential cutting to enable one to charge the total
cost of the edges to the LP solution. This simplifies the analysis and
also makes it transparent where one is using an algorithm designed for
general digraphs versus where planarity is important (and can be
potentially exploited in the future for further improvements). The
deterministic region growing lemma is easier to generalize to the
node-weighted setting. We note that the algorithm in \cite{KS22} is also randomized
because of a series of random shifts needed in the first layering 
decomposition step. There is no straightforward way to
derandomize this. Instead, we use deterministic region growing for
this step. Although we lose a logarithmic factor extra in our
procedure, this does not affect the final approximation factor.

\medskip
\noindent {\bf Discussion on quasi-metrics.}
Kawarabayashi and Sidiropoulos \cite{KS22} cast their rounding
algorithm in the language of probabilistic quasi-metric
embeddings and black-box the connection to Multicut and Sparsest Cut. 
This notion was developed
\cite{MemoliSS18,salmasi_constant_2019} in order to generalize the
metric embedding machinery to directed graphs. While this is a useful
perspective, we believe that this adds an extra layer of notational
overhead in the context of the results of \cite{KS22}.
To justify our
comment, we observe the following. Suppose we are given an
\emph{arbitrary} approximation algorithm for Multicut via the LP with
an approximation ratio of $\alpha$. Then, via LP duality, one can
argue the following: given any feasible LP solution ${\bf x}$, there
is an explicit and efficiently computable probability distribution
over feasible integral Multicut solutions such that sampling a
solution from the distribution cuts any fixed edge $e \in E$ with
probability at most $\alpha \cdot x_e$. This is standard and goes via
the arguments outlined in the work of Carr and Vempala
\cite{CarrV}. Thus, establishing integrality gaps for the Multicut LP
is the same as obtaining quasi-metric partitioning schemes --- in the
undirected settings this corresponds to the connection between
Multicut approximations and low-diameter decompositions. The embedding
machinery is particularly important and powerful when one is working
with multiple distance scales simultaneously. This is
needed if one wants to obtain tight approximations for the Sparsest
Cut problem. For instance, the $\ell_1$-embedding result of Bourgain
enabled the Sparsest Cut approximation in undirected graphs to be
improved from $O(\log^2 k)$ to $O(\log k)$
\cite{LLR,AumannR98}. However, in \cite{KS22}, the Sparsest Cut
approximation is obtained via the Multicut approximation
via a simple reduction that uses geometric grouping
over a logarithmic number of scales. For this reason, we adopt a direct approach
and describe our approximation algorithm for Multicut as 
a randomized rounding procedure that
outputs a feasible solution in which each vertex $v$ is deleted
with probability $O(\log^2  n) x_v$, where $x_v$ is the
fractional value assigned to $v$. We hope that our extension of the result in
\cite{KS22} to node-weights helps delineate the ideas that
may be needed to improve the edge-weighted case. 
In this context, we note
that some results for edge-weighted Sparsest Cut in undirected graphs,
such as that of Rao \cite{Rao} for planar and minor-free graphs, do
not yet extend to node-weighted Sparsest Cut.

\subsection{Other Related Work}
\label{sec:rel_work}
There is a vast literature on multicommodity flows and cuts, their
applications, and graph partitioning problems. The earlier discussion outlines 
some basic results, though they are many recent ones. 
We have mainly discussed algorithms for Multicut and
Sparsest Cut via flows. However, there are several other relevant techniques.

SDP-based algorithms for yield improved approximations
for Sparsest Cut. For the uniform case in undirected graphs, Arora, Rao
and Vazirani \cite{ARV} obtained an $O(\sqrt{\log n})$-approximation,
and for the non-uniform case, Arora, Lee and Naor \cite{ALN} obtained an
$O(\sqrt{\log n} \log \log n)$-approximation.  SDP-based
algorithms and analysis have been extended to several settings,
including to uniform sparsest cut in directed graphs by Agarwal et
al. \cite{AgarwalCMM05}; in particular, there is an
$O(\sqrt{\log n})$-approximation. Sparsest Cut in
undirected planar graphs admits
a quasi-polynomial time $(2+\eps)$-approximation
\cite{CohenAddadGKL21}
based on a combination of several tools, including the use
of a Sherali-Adams style linear programming relaxation.

Spectral methods based on Cheeger's inequality yield constant factor
approximations for conductance and expansion when the value is
large. Recent work has addressed spectral methods for conductance and
expansion in directed graphs as well. We refer the reader to a recent
paper \cite{Lau23} and the manuscripts of Spielman and Lau and others
on spectral graph theory.

In the context of directed graphs, an important line of work on
Multicut and Sparsest Cut is for the setting of ``symmetric'' demands.
Here the input is a directed graph $G=(V,E)$ and the pairs are sets
$\{s_1,t_1\},\ldots,\{s_k,t_k\}$. Separating a pair corresponds to
ensuring that $s_i$ and $t_i$ are not in the same strongly connected
component. A canonical problem in this setting is the Feedback Arcset
problem, where the goal is to remove a min-cost subset of the given
graph to make it acyclic. Polylogarithmic approximations and flow-cut
gaps are known for symmetric demand instances in directed graphs.  We
refer the reader to \cite{KPRT97,ENRS}. Feedback arcset and
Feedback vertex set in planar digraphs admit an 
$O(1)$-approximation \cite{BermanY,Sun24}.

We already mentioned that Multicut and Sparsest Cut in directed graphs
are hard to approximate to almost polynomial factors
\cite{ChuzhoyK}. Prior to this hardness result, there were several
results that obtained
polynomial-factor upper bounds \cite{CheriyanKR05,Gupta03,
  AgarwalAC07}. There have also been some positive results via
quasi-metric embeddings for special classes of graphs \cite{MemoliSS18,salmasi_constant_2019}.
Multicut in node-weighted planar graphs is at least as
hard as Vertex Cover in undirected graphs --- this implies
APX-Hardness and hardness of $(2-\eps)$ under UGC. As far as we know,
no super-constant factor hardness is known for Multicut in planar
digraphs and no APX-hardness is known for the Sparsest Cut problem.

\section{Preliminaries}
\label{sec:prelim}

\paragraph{Node-weighted Multicut in planar digraphs and the LP relaxation.}
We are given as input a directed planar 
graph $G = (V, E)$, along with a cost function on the vertices 
$c: V \to \R_{\geq 0}$ and terminal pairs $\{(s_i, t_i)\}_{i \in [k]}$.\footnote{As 
mentioned earlier, we can assume
all terminals have infinite cost to ensure they cannot be included in $S$.}
The goal is to find a minimum cost subset $S \subseteq V$ 
such that 
all terminal pairs are separated in $G \setminus S$; i.e. for all $i \in [k]$,
there is no $s_i$-$t_i$ path in $G \setminus S$.
We describe an LP relaxation for this problem. For each node $v \in V$, we have a variable 
$x_v$ denoting whether $v \in S$ or not. For each terminal pair, we let 
$\calP_i$ denote the set of all $s_i$-$t_i$ paths in $G$. 

\begin{align*}
  \min \sum_{v \in V} c(v) x_v& \\
  \text{subject to } \sum_{v \in p} x_v &\geq 1 &\forall p \in \calP_i, \forall i \in [k] \\
  x_v &\geq 0 &\forall v \in V
\end{align*}

Note that although this LP relaxation has exponentially many constraints, it 
can be solved efficiently via the ellipsoid method. 

\paragraph{Notation.} 
Let $G = (V, E)$ be a directed graph with node costs $c(v)$ and node lengths 
$x_v$. For any subgraph $H \subseteq G$, we let $V(H)$ and $E(H)$ denote the vertex 
and edge sets of $H$ respectively. For any path $p \subseteq G$, the \emph{length}
of $p$ is $\sum_{v \in p} x_v$. 
We let $d_x$ denote the distance
function given by the lengths $x$; that is, 
$d_x(u,v) = \min_{uv \text{ paths } p} \sum_{w \in p} x_w$. 
We drop $x$ when it is clear from context.
For any $v \in V$ and radius $r > 0$, we define the \emph{out-ball} centered at $v$ 
as $B^+(v, r) := \{u: d(v, u) \leq r\}$ and the \emph{in-ball} centered at $v$ as 
$B^-(v, r) := \{u: d(u,v) \leq r\}$. For a set $S$, we define the 
\emph{out-boundary}
as $\Gamma^+(S) = \{u \notin S: \exists v \in S \text{ s.t. } (v,u) \in E\}$.
Similarly, the \emph{in-boundary} of $S$ is 
$\Gamma^-(S) = \{u \notin S: \exists v \in S \text{ s.t. }
(u,v) \in E\}$. 
We define the \emph{volume} of a set as $\vol(S) = \sum_{v \in S} c(v)x_v$.
For any $i \geq 1$, we let $[i]$ denote the set $\{1, \dots, i\}$.

For any path $P = u_0, \dots, u_\ell$, we let $P[u_i, u_j]$ denote the 
subpath from $u_i$ to $u_j$. We call a path \emph{nontrivial} if $|P| \geq 2$.
For two nodes $u_i, u_j$ on a path $P$, we say 
$u_i \leq_{P} u_j$ if $u_i = u_j$ or $u_i$ appears before $u_j$ in $P$
(that is, $i \leq j$). For two paths $P, P'$, we say $P$ and $P'$ intersect 
if they share a common node, i.e. $V(P) \cap V(P') \neq \emptyset$.

\paragraph{Deterministic Region Growing.}
\label{sec:gvy_region}
A key tool used in this paper is a deterministic region growing lemma 
of Garg, Vazirani, and Yannakakis \cite{GVY}. Although this lemma 
was originally developed for edge-costs, it is well known that it extends 
to node-costs and directed graphs as well \cite{GVY-node,ENRS,CheriyanKR05}. For completeness, 
we provide a full proof in Appendix \ref{sec:gvy_proof}.

\begin{restatable}{lemma}{gvy}
\label{lem:gvy_region}
  Let $G = (V, E)$ be a directed graph with node costs $c(v)$ 
  and node lengths $x_v$.
  Suppose $x_u \leq \frac{\delta}{6\log n}$ for all $u \in V$.
  Fix $v \in V$ and a parameter $\delta > 0$. 
  There exists a radius $r \in [\frac\delta{6\log n},\delta)$ such that 
  \[c(\Gamma^+(B^+(v, r))) \leq O\left(\frac {\log n}{\delta}\right) 
  [\vol(B^+(v, r)) + \vol(V)/n].\] 
  Furthermore, there is an efficient deterministic algorithm to find $r$. 
  The same holds for the in-ball case.
\end{restatable}

We refer to the algorithm to obtain such a ball as $\RegionGrow^+(v, \delta)$ 
in the out-ball case and $\RegionGrow^-(v, \delta)$ in the in-ball case.
\section{The LP Rounding Algorithm}
\label{sec:algo}

We prove \Cref{thm:main} via a 
deterministic LP rounding algorithm. Given a feasible LP solution 
$\{x_v\}$, the algorithm consists of three main steps, 
following the same high level approach as \cite{KS22} with some 
modifications. 
We view $x_v$ as node lengths and aim to cut all long paths.
We first preprocess the solution to remove all $v$ with 
$x_v \geq \frac 1 {c \log n}$, where $c$ is a fixed constant. 
We introduce a parameter $\delta > 0$ below which we treat as 
a fixed constant. The exact values for $c, \delta$ will be specified 
later in the analysis.

\begin{definition}
  We say a digraph $G = (V, E)$ with distance function $d$ 
  is \emph{$\delta$-in-bounded} (similarly, 
  \emph{$\delta$-out-bounded}) if there exists a root $v \in V$ such 
  that for all $u \in V$, $d(u,v) \leq \delta$ (or, in the out-bounded 
  case, $d(v, u) \leq \delta)$. 
  We simply write $G$ is $\delta$-bounded if it is $\delta$-in-bounded 
  or $\delta$-out-bounded.
\end{definition}

\paragraph{Layering decomposition:} We partition
the graph $G$ into layers that are each $\delta$-bounded.
Furthermore, we find a low-cost set $S' \subseteq V$ such that all 
paths in $G \setminus S'$ are contained in at most two consecutive 
layers. This allows us to handle each layer independently.
Note that this decomposition preserves planarity, since each 
layer is a minor of $G$. However, the approach does 
not rely on planarity and in fact holds for general digraphs.
This is summarized in \Cref{lem:layering_main}.

\paragraph{Divide-and-conquer via planar separators:} 
For each layer $G_i$, since $G_i$ is $\delta$-bounded, there is a root $r_i$ 
that can reach (or can be reached by) all of $G_i$ within distance $\delta$.
This allows us to apply the planar separator lemma of Thorup 
\cite{Thorup} (see \Cref{lem:thorup_undir}) to obtain a separator consisting 
of three directed paths, each of length at most $\delta$. 
We cut regions around each of these paths, as specified in the following 
step, and recurse on each weakly connected component.

\paragraph{Cutting regions around paths:} We cut regions around each 
path in the planar separator. Specifically, for each separator path 
$P$, we obtain a set $S_P$ such that all long paths intersecting $P$ are 
cut by $S_P$; see \Cref{lem:one_path_main} for a formal statement.
As with the first step, this argument is not specific 
to planar graphs and would extend generally to any digraph.

\subsection{Constructing Layers}
\label{sec:layering}

In this section, we prove the following decomposition lemma.

\begin{lemma}
\label{lem:layering_main}
  Let $G = (V, E)$ be a weakly connected digraph with node costs $c(v)$ and 
  node lengths $x_v$. Let $\delta > 0$ be a given parameter 
  such that $x_v \leq \frac \delta {6 \log n}$ for all $v \in V$.
  There exists an efficient deterministic algorithm that decomposes $G$ 
  into $G_0, \dots, G_t$ for some $t \leq 2n$ such that:
  \begin{enumerate}
    \item Each $G_i$ is a $\delta$-bounded minor of $G$;
    \item $\sum_{i} \vol(V(G_i)) \leq \vol(V)$.
  \end{enumerate}
  Furthermore, this algorithm finds a set $S \subseteq V$ such that 
  $c(S) \leq O\left(\frac {\log n}{\delta}\right) \vol(V)$, and 
  for any path $P$ in $G \setminus S$, $\exists i \in \{0, \dots, t-1\}$ 
  such that $P$ can be decomposed into $P_0 \circ P_1$
  with $P_0 \subseteq G_i$ and $P_1 \subseteq G_{i+1}$.
\end{lemma}

\paragraph{Overview:} We follow the high-level approach of a 
decomposition given by Thorup \cite{Thorup}, which is as follows:
let $v_0 \in V$ be an arbitrary root, and
let $L_0$ denote the set of all $u \in V$ that are reachable from $v_0$ 
within some distance $\delta$. For $i \geq 1$,
\begin{itemize}[nosep]
  \item If $i$ is \emph{odd}, let $L_i$ denote the set of all
  $u \in V \setminus \cup_{j < i} L_j$ that \emph{can reach} $\cup_{j < i} L_j$
  in distance $\delta$;
  \item If $i$ is \emph{even}, let $L_i$ denote the set of all 
  $u \in V \setminus \cup_{j < i} L_j$ that \emph{are reachable from} 
  $\cup_{j < i} L_i$ in distance $\delta$.
\end{itemize}
This decomposition has a few nice properties. First, each layer is 
$\delta$-bounded. Second, we can remove the 
boundaries of these layers in a way that all remaining paths in $G$ are confined 
to a small number of consecutive layers.
In order to use this decomposition while 
ensuring that the cost of the removed 
boundary edges remained low, \cite{KS22} designed a randomized version of this decomposition
by adding a random shift at each layer. 

We instead use the deterministic region 
growing argument discussed in \Cref{sec:gvy_region}. The layer construction 
is formally described in \Cref{algo:layering}, see
\Cref{fig:layers}. 

\begin{figure}
  \centering
  \includegraphics[width=\linewidth]{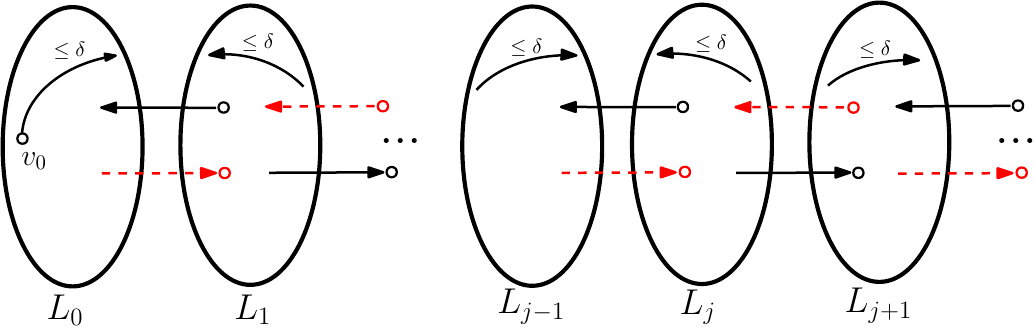}
  \caption{Example of consecutive layers constructed by \Cref{algo:layering},
  here $j$ is odd.
  Some boundary edges are shown with corresponding boundary nodes. The red 
  boundary nodes are the ones included in $S$; thus the corresponding incident 
  edges (shown with dashed red lines) are removed in $G \setminus S$. 
  }
  \label{fig:layers}
\end{figure}

\begin{algorithm}[H]
\caption{Constructing Layers in $(G, c, x), \delta$}
\label{algo:layering}
  \begin{algorithmic}
    \State $v_0 \gets$ arbitrary root in $V$ 
    \State $i \gets 0$
    \While{$|V(G)| \geq 1$}
      \If{$i$ is even}
        \State $L_i \gets \RegionGrow^+(v_i, \delta) \setminus \{v_i\}$
        \State $S_i \gets \Gamma^+(L_i)$
      \Else 
        \State $L_i \gets \RegionGrow^-(v_i, \delta) \setminus \{v_i\}$
        \State $S_i \gets \Gamma^-(L_i)$
      \EndIf
    \State Contract $L_i$ into $v_{i}$ to get $v_{i+1}$, with 
    $c(v_{i+1}) = x_{v_{i+1}} = 0$
    \State $i \gets i+1$
    \EndWhile\\
    \Return $L_0, \dots, L_{i}$, and $S := \cup_j S_j$
  \end{algorithmic}
\end{algorithm}

\begin{claim}
\label{claim:layering_iterations}
  Algorithm \ref{algo:layering} terminates  and at termination $i \leq 2n$.
\end{claim}
\begin{proof}
  We argue that for all $i$, if the algorithm has not terminated, then 
  there exists at least one node in 
  $(L_{i} \cup L_{i-1}) \setminus L_{i-2}$. This suffices to prove the claim. 
  Consider iteration $i-1$. Since the algorithm has yet to terminate,
  there must be some node $u \in V(G)$ such that $u \neq v_{i-1}$. 
  Since we assume $G$ is weakly connected, $u$ must be reachable from 
  $v_{i-1}$ or vice versa. Suppose $u$ is reachable from $v_{i-1}$; the 
  other case is analogous. Let $P$ be a path from $v_{i-1}$ to $u$, and let 
  $u'$ denote the first node after $v_{i-1}$ on $P$. Then,
  $(v_{i-1}, u') \in E$, so 
  $d(v_{i-1}, u') = x_{u'} \leq \delta/(6\log n)$. If $i-1$ is even, 
  then $u' \in L_{i-1}$; this is because 
  $\RegionGrow(v_{i-1}, \delta)$ returns an out-ball of radius at least 
  $\delta/(6\log n)$. Else, $i$ must be even and 
  since $v_{i-1}$ is contracted into $v_i$,
  $u' \in L_i$. In either case, $u' \in (L_{i} \cup L_{i-1}) \setminus 
  L_{i-2}$.
\end{proof}

\begin{lemma}
\label{lem:layering_cap} 
  $c(S) \leq O\left(\frac {\log n}{\delta}\right) \vol(V)$.
\end{lemma}
\begin{proof}
  By the guarantees given by \Cref{lem:gvy_region},
  $c(S) = \sum_{i} c(S_i) \leq \sum_{i} O\left(\frac {\log n}{\delta}\right) 
  [\vol(L_i) + \vol(V)/n] = O\left(\frac {\log n}{\delta}\right) 
  [\sum_i \vol(L_i) + \sum_i \vol(V)/n]$.
  We note here that $V(G)$ changes throughout 
  the course of the algorithm, but it is always decreasing as contracted 
  nodes $v_i$ add no volume. Thus, this is a valid upper bound. 

  For all $u \in V$, if $u \in L_i$, then $u$ is contracted 
  into $v_{i+1}$ and thus does not contribute any volume to $L_j$ for $j > i$. 
  Therefore, $\sum_i \vol(L_i) \leq \vol(V)$. Furthermore, by 
  \Cref{claim:layering_iterations}, the total number of layers is at 
  most $2n$. 
  Therefore,
  $O\left(\frac {\log n}{\delta}\right) 
  [\sum_i \vol(L_i) + \sum_i \vol(V)/n] \leq 
  O\left(\frac {\log n}{\delta}\right)[3 \cdot \vol(V)]$.
\end{proof}

\begin{claim}
\label{claim:layering_separation}
In $G \setminus S$, if $i$ is even, then
no node in $L_i$ can reach $(V \setminus L_i)$. If $i$ is odd, then no node 
in $L_i$ can be reached by $(V \setminus L_i)$. 
\end{claim}
\begin{proof}
  Let $i$ be even. Suppose for the sake of contradiction that there 
  exists a path in $G \setminus S$ from $L_i$ to $V \setminus L_i$. 
  Then, there must be some edge $(u,w)$ such that $u \in L_i$, $w \notin L_i$. 
  First, suppose $w \in L_j$ for $j < i$. Then either $w \in L_{i-1}$ 
  or $w$ was contracted into $v_{i-1}$. Thus 
  $u \in \Gamma^-(L_{i-1}) \subseteq S$ since $i-1$ is odd. Else, 
  $w \in L_j$ for $j > i$. Then since $u \in L_i$, 
  $w \in \Gamma^+(L_i) \subseteq S$ since $i$ is even. In either case, 
  this path contains a node in $S$, a contradiction. 

  Next, let $i$ be odd. Suppose for the sake of contradiction that there 
  exists a path in $G \setminus S$ from $V \setminus L_i$ to $L_i$. 
  Then, there must be some edge $(u,w)$ such that $u \notin L_i$, $w \in L_i$. 
  First, suppose $u \in L_j$ for $j < i$. Then either $u \in L_{i-1}$ 
  or $u$ was contracted into $v_{i-1}$. Thus 
  $w \in \Gamma^+(L_{i-1}) \subseteq S$ since $i-1$ is even. Else, 
  $u \in L_j$ for $j > i$. Then since $w \in L_i$, 
  $u \in \Gamma^-(L_i) \subseteq S$ since $i$ is odd. In either case, 
  this path contains a node in $S$, a contradiction.  
\end{proof}

\begin{lemma}
\label{lem:layering_paths}
  Let $P$ be any path in $G \setminus S$. Then, there exists some $i \geq 0$ 
  such that $P \subseteq L_i \cup L_{i+1}$. Furthermore, $P$ can be decomposed 
  into two subpaths $P_{i} \subseteq L_i$ and
  $P_{i+1} \subseteq L_{i+1}$.\footnote{This argument differs slightly from 
  that of \cite{KS22}.
  \cite{KS22} argues that every path of length at most $\delta$ is 
  contained in at most 3 layers. We generalize this to \emph{every} path 
  in order to reduce complications in later parts of the rounding scheme.}
\end{lemma}
\begin{proof}
  We denote the nodes on $P$ as $u_0, \dots, u_\ell$.
  Let $i$ be the minimum index such that $L_i \cap P \neq \emptyset$, and let 
  $u_j$ be some node in $L_i \cap P$. 
  
  Suppose $i$ is odd. Then, by \Cref{claim:layering_separation}, $P[u_0, u_j]$
  must be contained in $L_i$, since $u_j$ can be reached by all nodes in
  $P[u_0, u_j]$. If $P[u_j,u_\ell] \subseteq L_i$, we are done. Else, 
  let $j'$ be the first index such that $u_{j'} \notin L_i$. 
  Note that $d(L_i, u_{j'}) = x_{u_{j'}} \leq \delta/(6 \log n)$. 
  By construction, $L_{i+1}$ contains all nodes reachable within 
  $\delta/(6 \log n)$ from $L_i$. Thus $u_{j'} \in L_{i+1}$. By 
  \Cref{claim:layering_separation}, $u_{j'}$ cannot reach any node 
  in $V \setminus L_{i+1}$. Thus $P[u_{j'}, u_{\ell}] \subseteq L_{i+1}$.
  In this case, $P_i = P[u_0,u_{j'-1}]$ and $P_{i+1} = P[u_{j'}, u_{\ell}]$.
  
  The case where $i$ is even is similar. 
  By \Cref{claim:layering_separation}, $P[u_j, u_\ell]$
  must be contained in $L_i$.
  If $P[u_0, u_j] \subseteq L_i$, we are done. Else, 
  let $j'$ be the largest index such that $u_{j'} \notin L_i$. 
  Then $d(u_{j'}, L_i) \leq x_{u_{j'}} \leq \delta/(6 \log n)$. 
  By construction, $L_{i+1}$ contains all nodes that can reach $L_i$
  within distance $\delta/(6 \log n)$. Thus $u_{j'} \in L_{i+1}$, and by 
  \Cref{claim:layering_separation}, 
  $P[u_0, u_{j'}] \subseteq L_{i+1}$.
  In this case, $P_{i} = P[u_{j'-1}, u_{\ell}]$ and $P_{i+1} = P[u_0,u_{j'}]$.
\end{proof}

\begin{proof}[Proof of \Cref{lem:layering_main}]
  We call \Cref{algo:layering} to obtain $S$ and layers 
  $L_i$ for $i = 0, \dots, t$. By \Cref{claim:layering_iterations},
  $t \leq 2n$.
  Since \RegionGrow~is deterministic and efficient, 
  \Cref{algo:layering} is as well. 
  We define $G_0$ as $L_0$ and 
  for each $i \in [t-1]$, we define $G_i$ as the minor of $G$ 
  obtained by contracting $\cup_{j < i} L_j$ into a new node $v_i$
  \footnote{Note that $v_i$ here corresponds exactly to $v_i$ as defined 
  in \Cref{algo:layering}.}
  and deleting $\cup_{j > i} L_j$. We set $c(v_i) = x_{v_i} = 0$ for all 
  $i$. 
  By construction, each $G_i$ is a minor of $G$ and is $\delta$-bounded 
  with root $v_i$. Furthermore, since the layers $L_i$ are disjoint, 
  $\sum_i \vol(V(G_i)) = 
  \sum_i \sum_{v \in L_i} c(v)x_v = \sum_{v \in V} c(v)x_v = \vol(V)$.
  The properties of $S$ guaranteed by \Cref{lem:layering_main} are given 
  directly by \Cref{lem:layering_cap} and \Cref{lem:layering_paths}.
\end{proof}
\subsection{Planar Separators}
\label{sec:separator}

In this section, we further decompose each ``layer''
into a collection of paths and their corresponding neighborhoods.
We prove the following lemma.

\begin{lemma}
\label{lem:separator_main}
  Let $G = (V, E)$ be a $\delta$-bounded planar digraph with node costs $c(v)$ and 
  node lengths $x_v$ such that $x_v \leq \frac \delta {6 \log n}$ for all $v \in V$.
  Then, there exists an 
  efficient deterministic algorithm to find a set $S \subseteq V$
  such that $c(S) \leq O(\log^2n) \sum_{v \in V} c(v)x_v$,
  and for all $u, v \in V$, if $u$ can reach $v$ in $G \setminus S$, 
  then $d_x(u,v) \leq 3\delta$.
\end{lemma}

The main tool used in this section is a result on shortest path separators in 
planar digraphs. This is the key property of planarity used in this 
work.

\begin{lemma}[\cite{Thorup}]
\label{lem:thorup_undir}
  Let $H$ be an \emph{undirected} planar graph with a spanning tree 
  $T$ rooted at some node $r \in V(H)$. Then, there is an efficient 
  deterministic algorithm that finds $v_1, v_2, v_3 \in V(H)$ such that the 
  following holds: let $P_i$ be the unique $r$-$v_i$ path for $i \in [3]$. Then 
  each connected component of $H \setminus (P_1 \cup P_2 \cup P_3)$ has 
  at most half the vertices of $H$. 
\end{lemma}

By applying the above lemma on a shortest-path out-tree of any planar 
digraph, we immediately obtain the following corollary.

\begin{corollary}
\label{cor:thorup_dir}
  Let $H$ be a planar digraph with root $r \in V(H)$ such that $r$ can reach 
  all nodes in $V(H)$. 
  Then, there is an efficient deterministic 
  algorithm that finds three shortest paths $P_1, P_2, P_3$ starting at $r$
  such that each weakly connected component of 
  $H \setminus (P_1 \cup P_2 \cup P_3)$ has at most half the vertices of $H$. 
\end{corollary}

Note that we can obtain the same result for three shortest paths \emph{ending}
at $r$ by considering a shortest-path in-tree instead, or by 
reversing the directions of all edges and directly applying 
\Cref{cor:thorup_dir}.

\paragraph{Overview:} We use the planar separator to follow 
a recursive process. Consider a $\delta$-bounded planar digraph $G$. 
Since $G$ is $\delta$-bounded (we assume out-bounded for this 
description), there is a root $r$ that can reach all of $G$ within distance 
$\delta$. Let $P_1, P_2, P_3$ the separator given by \Cref{cor:thorup_dir}
on $G, r$, and note that each $P_i$ has length at most $\delta$.
Ideally, we would like to remove these paths and recurse on the remaining 
weakly connected components; since each component has at most half the vertices,
$G$ would be fully separated in $O(\log n)$ iterations. However, although these 
separators have short length, they may have high cost. Instead, we show that 
we can find low-cost cuts \emph{surrounding} each separator. In \Cref{sec:one_path} 
we describe the process of cutting regions around a separator, and in 
\Cref{sec:separator_proof} we formally describe this recursive process to 
prove \Cref{lem:separator_main}.

\subsubsection{Cutting Regions around Path}
\label{sec:one_path}

In this section, we cut small-diameter regions around a given 
path using the region growing algorithm discussed in \Cref{sec:gvy_region}.
We prove the following lemma.

\begin{lemma}
\label{lem:one_path_main}
Let $G = (V, E)$ be a digraph with node costs $c(v)$ and node lengths 
$x_v$. Let $P$ be a shortest path of $G$ of length at most $\delta > 0$. 
Suppose for all $v \in V$, $x_v \leq \frac {\delta}{6 \log n}$. 
Then, there exists an efficient deterministic algorithm to find a set 
$S \subseteq V$ such that 
$c(S) \leq O\left(\frac {\log n}{\delta}\right) \vol(V)$, and 
for all $u,v \in V$, if there exists some $u$-$v$ path $Q$ in 
$G \setminus S$ such that $Q$ intersects $P$, then $d(u,v) \leq 
3 \delta$.
\end{lemma}

Let $G = (V, E)$ be such a digraph with costs $c(v)$, lengths $x_v$, and 
a shortest path $P$.
We denote the nodes on $P$ as $v_0, \dots, v_\ell$. 
We define \PathCutting($G, P, \delta$) as follows:
\begin{itemize}[left=0pt]
  \item \textbf{Cutting out-balls:} We cut out-balls around $P$ starting at 
  the \emph{end} of the path. Let $\curr = \ell$, set 
  $B_1^+ := \RegionGrow^+(v_\curr, \delta)$ and $S_1^+ := \Gamma^+(B_1^+)$.
  We then remove $B_1^+$ and recurse on $G \setminus B_1^+$. At each iteration 
  $i$, given out-balls $B_1^+, \dots, B_{i-1}^+$, set 
  $\curr = \argmax_{j < \curr} \{v_j \notin B_{i-1}^+\}$. That is, 
  $v_{\curr}$ is the furthest remaining node on $P$ not contained in any 
  out-ball. Set $B_i^+ := \RegionGrow^+(v_\curr, \delta)$ and 
  $S_i^+ := \Gamma^+(B_i^+)$, and remove $B_i^+$. 
  Note that $B_i^+$ and $S_i^+$ are constructed with respect to 
  $G \setminus \cup_{j < i} B_j^+$ and are thus disjoint from all $B_j^+$ for 
  $j < i$. The process terminates when all nodes in $P$ are contained in 
  an out-ball. See \Cref{fig:out_balls}. 
  \begin{figure}
    \centering
    \includegraphics[width=\linewidth]{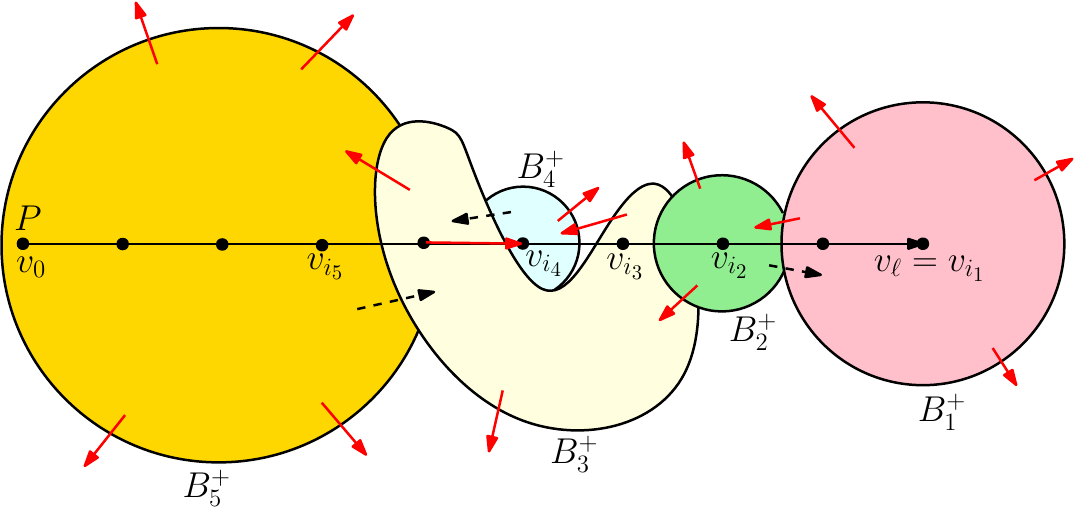}
    \caption{Example of out-ball construction. The centers $v_{\curr}$ of 
    the ball $B_j^+$ are denoted $v_{i_j}$. Note that the balls do not 
    necessarily cover contiguous portions of the path $P$, as shown by 
    $B_3^+$ and $B_4^+$. The boundary edges are depicted in red; boundary 
    vertices $S_i^+$ are not shown for simplicity. We highlight the fact that 
    edges from $B_j^+$ to $B_i^+$ (given by dashed black edges)
    for $j > i$ are not necessarily removed 
    in this process.}
    \label{fig:out_balls}
  \end{figure}
  \item \textbf{Cutting in-balls:} We follow a similar process as above, starting 
  instead from the \emph{beginning} of the path. Let $\curr = 0$ and let 
  $B_1^- := \RegionGrow^-(v_\curr, \delta)$ and $S_1^- := \Gamma^-(B_1^-)$. 
  We remove $B_1^-$ and recurse on $G \setminus B_1^-$. 
  At each iteration 
  $i$, given in-balls $B_1^-, \dots, B_{i-1}^-$, set 
  $\curr = \argmin_{j > \curr} \{v_j \notin B_{i-1}^-\}$ so
  $v_{\curr}$ is the closest node on $P$ not contained in any 
  in-ball. Set $B_i^- := \RegionGrow^-(v_\curr, \delta)$ and 
  $S_i^- := \Gamma^-(B_i^-)$. Note that $B_i^-$ and $S_i^-$ are with respect to 
  $G \setminus \cup_{j < i} B_j^-$ and are thus disjoint from all $B_j^-$ for 
  $j < i$. The process terminates when all nodes in $P$ are contained in 
  an in-ball.
  \item Return $S = \bigcup_i S_i^+ \cup \bigcup_i S_i^-$
\end{itemize}

We first bound the cost of $S$. 

\begin{lemma}
\label{lem:one_path_cost}
  Let $S$ be given by \PathCutting($G, P, \delta$). 
  Then $c(S) \leq O\left(\frac {\log n}{\delta}\right) \vol(V)$.
\end{lemma}
\begin{proof}
  We first bound the cost of $\cup_i S_i^+$. By \Cref{lem:gvy_region},
  $c(S_i^+) \leq O\left(\frac {\log n}{\delta}\right) 
  [\vol(B_i^+) + \vol(V)/n]$. 
  By construction, the $B_i^+$ are all disjoint; once we construct a ball, we 
  remove it from $G$, thus it does not get used in future balls or boundaries. 
  Therefore, 
  $\sum_i \vol(B_i^+) \leq \vol(V)$. Furthermore, the total number of 
  out-balls is at most $n$, since each out-ball covers at least one 
  new node on $P$. Thus 
  $c(\cup_i S_i^+) = \sum_i c(S_i^+) = O\left(\frac {\log n}{\delta}\right)  
  \vol(V)$ as desired. A similar argument shows that 
  $c(\cup_i S_i^-) = O\left(\frac {\log n}{\delta}\right)  
  \vol(V)$, giving us our desired bound on $c(S)$. 
\end{proof}

Next, we prove that $S$ cuts long paths that intersect with $P$. This, 
along with \Cref{lem:one_path_cost}, suffices to prove 
\Cref{lem:one_path_main}.

\begin{claim}
\label{claim:one_path_reachability}
  Let $\{B_i^+\}_i$ and $\{B_i^-\}_i$ denote the out-balls and in-balls
  constructed in \PathCutting($G, P, \delta$). Fix $i$, and suppose 
  $u \in B_i^-$. Then for all $w \in V$, if $u$ is reachable from 
  $w$ in $G \setminus S$, then $w \in \cup_{j \leq i} B_{j}^-$. 
  Similarly, if $u \in B_i^+$, then for all $w \in V$, if $u$ 
  can reach $w$ in $G \setminus S$, then $w \in \cup_{j \leq i} B_j^+$.
\end{claim}
\begin{proof}
  We prove the claim for in-balls; the argument for out-balls is analogous. 
  Fix $i$. 
  For ease of notation, we write $\calB_i$ to denote $\cup_{j \leq i} B_j^-$. 
  We will show that no node in $\calB_i$ is reachable from $V \setminus \calB_i$;
  this suffices to prove the claim.
  Let $u \in \calB_i$, 
  $w \in V \setminus \calB_i$, and consider any 
  $w$-$u$ path $Q$ in $G$. We will show that $Q$ must contain at least one node of $S$. 
  There must be some edge $(w', u')$ in $Q$ going from 
  $V \setminus \calB_i$ to $\calB_i$. 
  Let $j \leq i$ be such that $u' \in B_j^-$. Since $w' \in V \setminus \calB_i$,
  $w' \in V(G)$ at the time that $B_j^-$ is constructed. Thus 
  $w' \in S_j^- \subseteq S$ as desired.
\end{proof}

\begin{lemma}
\label{lem:one_path_cuts} 
  Let $S$ be given by \PathCutting($G, P, \delta$). 
  For all $u,v \in V$, if there exists some $u$-$v$ path $Q$ in 
  $G \setminus S$ such that $Q$ intersects $P$, then $d(u,v) \leq 
  3 \delta$.\footnote{This proof follows the same structure as 
  that of \cite{KS22} but has been significantly simplified.}
\end{lemma}
\begin{proof}
  For any node $u \in \cup_i B_i^-$, we let $v_{in}(u)$ denote the center of the 
  in-ball containing $u$; note that this is unique since the balls 
  $B_i^-$ are disjoint. For any $u \in \cup_i B_i^+$, we let $v_{out}(u)$ 
  denote the center of the out-ball containing $u$. Note that by construction, 
  if $u \in P$ then
  $v_{in}(u) \leq_P u$ and $u \leq_P v_{out}(u)$, since in-balls are centered at 
  the first remaining node in $P$ and out-balls are centered at the last remaining 
  node in $P$.

  Fix $u, w \in V$. Let $Q$ be any $u$-$w$ path in $G \setminus S$ that intersects $P$. 
  Let $a, b$ denote the first and last nodes in $Q \cap P$ when traversing 
  $Q$ from $u$ to $w$. Note that $a$ may appear before or after $b$ in $P$.
  Since $a, b \in P$, they both must be contained in $\cup_i B_i^-$ and in 
  $\cup_i B_i^+$. By \Cref{claim:one_path_reachability}, since $u$ can reach $a$ in 
  $G \setminus S$, $u \in \cup_i B_i^-$ and $v_{in}(u) \leq_P v_{in}(a)$. 
  Similarly, since $b$ can reach $w$ in $G \setminus S$,
  $w \in \cup_i B_i^+$ and $v_{out}(b) \leq_P v_{out}(w)$ 
  (recall that the index of the out-ball is smaller when the center is further, since 
  the construction of out-balls starts at the end).
  Since $a$ can reach $b$ in $G \setminus S$, $v_{in}(a) \leq_P v_{in}(b)$. 
  See \Cref{fig:path_proof}.
  \begin{figure}
    \centering
    \includegraphics[width=\linewidth]{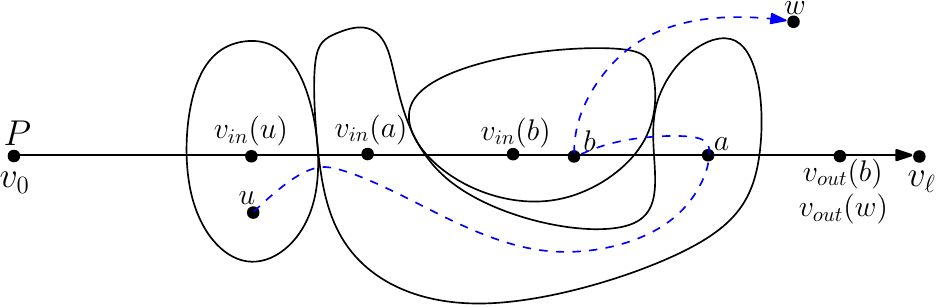}
    \caption{Here $Q$ is a $u$-$w$ path shown with a dashed blue line. All 
    regions drawn are in-balls; we omit the out-balls for visual clarity as 
    the analysis is similar. 
    In this example, $b \leq_P a$ to illustrate that $a$ and $b$ can appear 
    in any order, even though $v_{in}(a) \leq_P v_{in}(b)$ necessarily.}
    \label{fig:path_proof}
  \end{figure}
  Thus 
  \[v_{in}(u) \leq_P v_{in}(a) \leq_P v_{in}(b) \leq_P b \leq_P v_{out}(b)
  \leq_P v_{out}(w).\]
  In particular, since the length of $P$ is at most $\delta$, 
  $d_x(v_{in}(u), v_{out}(w)) \leq \delta$. By construction of the in and out balls,
  $d_x(u, v_{in}(u)) \leq \delta$ and $d_x(v_{out}(w), w) \leq \delta$. 
  Thus $d_x(u, w) \leq 3\delta$.
\end{proof}

\begin{remark}
  Cutting the balls sequentially in this manner is necessary because we are 
  working with directed graphs. For instance, suppose there exists some $u \in V$,
  $i > 0$ such that $d(u, v_i)$ is small (at most $\delta$) and 
  $d(u, v_0)$ is large (at least $3\delta$). This is possible since $P$ is directed.
  By the ball cutting procedure, if $B_i^-$ is constructed before 
  $B_j^-$, then all paths from $B_j^-$ to $B_i^-$  are cut, but such a guarantee 
  does not exist for long paths from $B_i^-$ to $B_j^-$.
  Then, cutting in-balls sequentially \emph{starting} at $v_0$ is the only way 
  to ensure that $u$ can no longer reach $v_0$ in $G \setminus S$. 
  The reasoning for out-balls is similar.
\end{remark}

\subsubsection{Recursive Process}
\label{sec:separator_proof}

\begin{proof}[Proof of \Cref{lem:separator_main}]
  Let $G = (V, E)$ be a $\delta$-bounded planar digraph with root $v$ and 
  cost and length functions $c, x$ on $V$. 

  The algorithm is as follows. We apply \Cref{cor:thorup_dir} to obtain 
  $P := P_1 \cup P_2 \cup P_3$; note that each of these paths has length 
  at most $\delta$. For $i = 1, 2, 3$, let
  $S_i = \PathCutting(G, P_i, \delta)$ and set $S_G = \cup_{i \in [3]} S_i$. 
  We then contract $P$ into $v$, and for each weakly connected component 
  $C$ of $G \setminus P$, recurse on $C' := C \cup \{v\}$, with 
  $c(v) = x_v = 0$.
  Note that $C'$ 
  is a minor of $G$ and is thus planar. Furthermore, $C'$ is 
  $\delta$-bounded with root $v$ since contracting $P$ into $v$ does not 
  increase distances from $v$. At each recursive call on the connected component 
  $C$, we let $S_C$ denote the resulting cut set. We proceed recursively 
  until each remaining weakly connected component consists of at most 
  one vertex. Since the number of vertices are halved at each step, 
  there are at most $\lceil \log n \rceil$ levels of recursion. For 
  $t = 0, \dots, \lceil \log n \rceil$, we let $\calC_t$ denote the set 
  of all weakly connected components at the $t^{th}$ level of recursion. 
  We return $S := \cup_{t} \cup_{C \in \calC_t} S_C$.

  It is clear that this algorithm is deterministic and efficient, since 
  \PathCutting, \RegionGrow, and the algorithm given by \Cref{cor:thorup_dir}
  are all deterministic and efficient. For the cost bound,
  by \Cref{lem:one_path_cost}
  $
    c(S) = \sum_t \sum_{C \in \calC_t} c(S_C)
    \leq \sum_t \sum_{C \in \calC_t} O\left(\frac {\log n}{\delta}\right) \vol(V(C)).
  $
  For $C, C'$ in the same level of recursion, the only node in 
  $V(C) \cap V(C')$ is the contracted root $v$. Since this has volume $0$,
  $\sum_{C \in \calC_t} \vol(V(C)) \leq \vol(V)$ for all $t$.  
  Thus $c(S) \leq \sum_t O\left(\frac {\log n}{\delta}\right) \vol(V)
  \leq O\left(\frac {\log^2 n}{\delta}\right) \vol(V)$. 

  For the path guarantee, let $u, v \in V$ and let $Q$ be any $u$-$v$ path 
  in $G \setminus S$. 
  At the end of the recursive process, $u$ and $v$ are not in the same 
  connected component. Let $C$ be the last component in which $Q \subseteq C$. 
  Since $Q$ is not fully contained in any components after $C$, it must be the 
  case that $Q$ intersects the corresponding planar separator $P$. 
  Since $\PathCutting(C, P, \delta) \subseteq S$, by \Cref{lem:one_path_cuts},
  $\delta(u,v) \leq 3\delta$. 
\end{proof}

\begin{remark}
  \label{rem:contraction}
  We mention a minor nuance of the above argument. At each recursive step after 
  the first, the contracted root is not a vertex in $G$, and thus cannot be 
  included in $S$, since we require that $S \subseteq V(G)$. To avoid this issue, 
  we can replace all calls $\PathCutting(C, P, \delta)$ with 
  $\PathCutting(C, P \setminus \{v\}, \delta)$ where $v$ is the root. This 
  does not affect the feasibility, since if any path $Q$ intersects the 
  separator $P$ at the contracted root, it must have intersected the separator 
  at an earlier level as well. Thus, for every path $Q$, there is some 
  component $C$ such that $Q \subseteq C$ and $Q$ intersects the corresponding 
  separator at a non-root vertex.
\end{remark}

\subsection{Proof of Main Theorem}

We prove our main result \Cref{thm:main}.
Let $G = (V, E)$ be a directed planar graph with node capacities 
$c: V \to \R_\geq 0$. We solve the LP relaxation for node-weighted 
Multicut to obtain an optimal LP solution $\{x_v\}_{v \in V}$. 
Note that $\vol(V)$ is exactly the value of the optimal LP solution,
which we denote $\opt_{LP}$.

\paragraph{Algorithm:} 
We assume $G$ 
is weakly connected; else, each weakly connected component of $G$ can be 
handled separately. 
Set $S_L = \{v: x_v \geq \frac \delta {6\log n}\}$.
We apply the layering decomposition 
\Cref{lem:layering_main} on $G \setminus S_L$ with $\delta = 1/12$
to obtain $S'$, along with 
a series of $\delta$-bounded minors $G_0, \dots, G_t$. Note that each $G_i$ 
is planar, since it is a minor of $G$. We then apply 
\Cref{lem:separator_main} on each $G_i$ to obtain a corresponding set $S_i$.
Return $S := S_L \cup S' \cup \cup_{i = 0}^t S_i$.

\paragraph{Cost:} The cost of $S_L$ is 
$c(S_L) \leq \frac{6\log n}{\delta} \sum_{v \in S_L} c(v)x_v \leq 
O(\log n) \vol(V)$. By \Cref{lem:layering_main}, the cost of $S'$ 
is $c(S') \leq O\left(\frac {\log n}{\delta}\right) \vol(V)
= O(\log n) \vol(V)$. By \Cref{lem:separator_main}, for each 
$i = 0, \dots, t$ $c(S_i) \leq O\left(\frac {\log^2 n}{\delta}\right)
\vol(V(G_i))$. By \Cref{lem:layering_main}, $\sum_{i = 0}^t \vol(V(G_i)) 
\leq \vol(V)$. Thus $c(S) \leq O(\log^2 n) \vol(V)$, giving us our 
desired approximation ratio.

\paragraph{Feasibility:} We show that all terminal pairs 
are separated in $G \setminus S$. Suppose for the sake of contradiction that 
some terminal pair $(s_i, t_i)$ is not separated. Let $P$ be an
$s_i$-$t_i$ path in $G \setminus S$. 
By \Cref{lem:layering_main}, for some $j \in \{0, \dots, t-1\}$, 
$P$ can be decomposed into $P_0 \circ P_1$ such that 
$P_0 \subseteq G_j$ and $P_1 \subseteq G_{j+1}$. Let $(u_0, u_1)$ be the edge 
at which $P_0$ and $P_1$ split; that is, $P_0$ is an $s_i$-$u_0$ path 
and $P_1$ is a $u_1$-$t_i$ path. Since $P \subseteq G \setminus S$,
$s_i$ can reach $u_0$ in $G_j \setminus S_j$. 
Thus by \Cref{lem:separator_main}, $d_x(s_i, u_0) \leq 3\delta = \frac 1 4$. 
Similarly, $d_x(u_1, t_i) \leq \frac 1 4$. This implies 
$d_x(s_i, t_i) \leq \frac 1 2$, contradicting feasibility of the LP.

\begin{remark}
  Following the discussion in \Cref{rem:contraction}, it is important that 
  the contracted roots of each $G_i$ are not included in $S$. 
  \Cref{lem:layering_paths} guarantees that all paths in $G \setminus S'$ 
  are contained in two consecutive layers; this does not include the contracted 
  roots. Thus, in the separator step, 
  it suffices to restrict attention to $u, v$ that are reachable in 
  $G_i \setminus v_i$. This does not change the argument, thus we omit 
  this detail from the above proofs for simplicity.
\end{remark}
\section{Conclusion}
\label{sec:conclusion}
We showed that the $O(\log^2 n)$-factor approximation algorithm of
Kawarabayashi and Sidiropoulos \cite{KS22} for Multicut in planar
digraphs can be extended to the node-weighted case. Are there better
approximations for Multicut and Sparsest Cut in planar digraphs?  What
are tight bounds on the the flow-cut gap?  Constant factor
approximations and flow-cut gaps are not ruled out. Extending the
results to digraphs supported by minor-free graphs is another natural
direction.

\bibliographystyle{plainurl}
\bibliography{references}

\appendix
\section{Proof of Node-Weighted Directed Region Growing}
\label{sec:gvy_proof}

\gvy*

\begin{proof}
  We prove the lemma statement for the out-ball case. The in-ball case can 
  be obtained directly by reversing the directions of all edges. 
  We define a series of balls centered at $v$: for 
  $i = 1, \dots, 3 + 2\lceil \log n\rceil$,
  let $r_i = i \cdot \delta/(6\log n)$ and consider $B_i := B^+(v, r_i)$. 
  Note that for all $i$, $r_i \in [\frac\delta{6\log n},\delta)$.

  First we show that for all $i$, if $u \in \Gamma^+(B_i)$, then 
  $u \in B_{i+1}$. Since $u \in \Gamma^+(B_i)$, there is some 
  $w \in B_i$ such that $(w, u) \in E$. Thus $d(v, u) \leq d(v, w) 
  + x_u \leq r_i + \frac \delta {6\log n} = r_{i+1}$. 
  
  Next, we show that there must be some $i$ such that $\vol(B_{i+2}) 
  \leq 2\vol(B_i) + \vol(V)/n$. If not, this would imply that 
  $[\vol(B_{i+2}) + \vol(V)/n] > 2[\vol(B_i) + \vol(V)/n]$ for all $i$. 
  In particular, 
  $[\vol(B_i) + \vol(V)/n] > 2^{(i-1)/2}[\vol(B_1) + \vol(V)/n]$, 
  so 
  $[\vol(B_{3+2\lceil log n\rceil}) + \vol(V)/n] > 2n[\vol(B_1) + \vol(V)/n]
  \geq 2\vol(V)$, a contradiction. Fix such an $i$. 

  Suppose we were to select $r$ uniformly at random from $[r_i, r_{i+1})$,
  and let $B_r$ denote $B^+(v, r)$. Consider the probability that some 
  $u \in V$ is contained in $\Gamma^+(B_r)$. By earlier analysis, this is 
  only possible if $u \in B_{i+2}$.
  Let $P$ be the shortest 
  $v$-$u$ path, and let $w$ be the node preceding $u$ on this path. We claim 
  that $u \in \Gamma^+(B_r)$ if and only if $w \in B_r$ and $u \notin B_r$. 
  One direction is simple; since $(w, u) \in E$, so $w \in B_r, u \notin B_r$
  implies $u \in \Gamma^+(B_r)$. For the other direction, suppose 
  $u \in \Gamma^+(B_r)$. Then there exists some $w' \in B_r$ such that 
  $(w', u) \in E$. Since $P$ is the shortest $v$-$u$ path, it must be the case 
  that $d(v, u) = d(v, w) + x_u \leq d(v, w') + x_u$. Thus 
  $d(v, w) \leq d(v, w') \leq r$, so $w \in B_r$.
  In particular, this implies that 
  \[\Pr[u \in \Gamma^+(B_r)] = \Pr[d(v, w) \leq r < d(v, u)] \leq 
  \frac{d(v, u) - d(v, w)}{r_{i+1}-r_i} = \frac{6\log n}{\delta} \cdot x_u.\]
  Thus 
  \begin{align*}
    E[c(\Gamma^+(B_r))] &\leq \sum_{u \in B_{i+2}} c(u)\Pr[u \in \Gamma^+(B_r)]
    = \frac {6\log n}{\delta}\sum_{u \in B_{i+2}} c(u)x_u =  
    \frac {6\log n}{\delta} \vol(B_{i+2}). 
  \end{align*}
  By choice of $i$, this is at most 
  $\frac{6\log n}{\delta} [2\vol(B_i) + \vol(V)/n]$. Since $B_i \subseteq B_r$,
  $\vol(B_i) \leq \vol(B_r)$, concluding the proof of existence of $r$.

  To find $r$ deterministically, one can check all relevant points between 
  $r_i$ and $r_{i+1}$, since there are at most $n$ values where $B_r$ changes.
\end{proof}
 
\end{document}